\documentclass[prl,showpacs,twocolumn,typeset,superscriptaddress]{revtex4-1} 

\usepackage{mathtools}
\usepackage{float}
\usepackage{graphicx}
\usepackage{epsfig}
\usepackage{subfloat}
\usepackage{dcolumn}
\usepackage{amsmath,amssymb,amsthm}
\usepackage{bm}
\usepackage{verbatim}
\usepackage{natbib}
\usepackage[normalem]{ulem}
\usepackage{color}

\newtheorem{theorem}{Theorem}
\newtheorem*{theoremn}{Theorem 1}
\newtheorem*{definition}{Definition}
\newcommand{\bra}[1]{\langle #1|}
\newcommand{\ket}[1]{|#1\rangle}

\newcommand{\ketbra}[2]{| #1 \rangle \langle #2 |}
\newcommand{\M}[1]{\mathcal{#1}}
\newcommand{\id}{\mathbb{I}}
\newcommand{\B}[1]{\mathbb{#1}}
\newcommand{\Tr}{\text{Tr}}
\newcommand{\N}[1]{\Vert #1 \Vert}

\usepackage{textcomp}

\begin{document}
\title{Non-Markovianity Quantifier of an Arbitrary Quantum Process} 
\author{Tiago Debarba} 
\affiliation{Universidade Tecnol\'ogica Federal do Paran\'a (UTFPR), Campus Corn\'elio Proc\'opio, Avenida Alberto Carazzai 1640, Corn\'elio Proc\'opio, Paran\'a 86300-000, Brazil}
\email{debarba@utfpr.edu.br}
\author{Felipe F. Fanchini}
\affiliation{Faculdade de Ci\^encias, UNESP - Universidade Estadual Paulista, Bauru, S\~ao Paulo 17033-360, Brazil}
\email{fanchini@fc.unesp.br}
\date{\today}
\begin{abstract}
Calculating the degree of non-Markovianity of a quantum process, for a high-dimensional system, is a difficult task given complex maximization problems. Focusing on the entanglement-based measure of non-Markovianity we propose a numerically feasible quantifier 
for finite-dimensional systems. We define the non-Markovianity measure in terms of a class of entanglement quantifiers named witnessed entanglement which allow us to write several entanglement based measures of non-Markovianity in a unique formalism. 
In this formalism, we show that the non-Markovianity, in a given time interval, can be witnessed by calculating the expectation value of an observable, making it attractive for experimental investigations. Following this property we introduce a quantifier base on the 
entanglement witness in an interval of time, we show that measure is a {\it bona fide} measure of non - Markovianity. In our example, we use the generalized robustness of entanglement, an entanglement measure that can be readily calculated by a semidefinite programming method, to study impurity atoms coupled to a bose-einstein condensate. 
\end{abstract}

\pacs{03.65.Yz, 42.50.Lc, 03.65.Ud}
\maketitle

\section{Introduction}

Recently, non-Markovian processes have been receiving special attention in the literature \cite{breuer16}. In fact, the understanding of non-Markovian process is important for practical applications such as communication protocols, quantum algorithms, and quantum cryptography \cite{chuang}. Much effort has been devoted to the characterization and quantification of non-Markovianity in dissipative process. Several measures have been presented, and among the most evident, we can mention the measure based on the trace distance defined by Breuer, Laine, and Pillo \cite{breuer09}, the measure based on entanglement dynamics defined by Rivas, Huelga, and Plenio \cite{plenioprl10}, and the measure based on the dynamics of mutual information defined by Luo, Fu, and Song \cite{luo12}.
% Rev. Mod. Phys

It is known that, in general, such measures do not agree on the description of dissipative dynamics, differing even when defining a process such as Markovian \cite{nm1,nm2,nm3, neto16}. Moreover, besides these differences, they have a particular and unpleasant feature in common: all are computationally difficult to evaluate because the number of variables involved is hugely increased with the system size. {In one hand, for the great majority of measure we have a maximization over all initial condition, while for the entanglement-based measurement (where the maximal entangled state is choose as default) we have the problem to calculate the amount of entanglement per se.}
This makes the study of non-Markovian processes an extremely difficult task, and it is exactly in this direction that we develop our work by presenting an efficient new strategy.

Focusing on the non-Markovianity measure based on entanglement dynamics, we introduce, for an arbitrary quantum process and independent of the system size, a quantifier of non-Markovianity based on the dynamics of the generalized robustness of entanglement. This entanglement quantifier can be written according to the entanglement witness approach \cite{brandaopra2005}, which allows its evaluation in an efficient manner through robust semidefinite programming (RSDP) \cite{rsdpreinaldo}. {In a sense, the witnessed entanglement approach generalize the usual entanglement-based measure of quantum non-Markovianity, since several entanglement quantifiers can be written in an unique formalism}. 
{Furthermore, we demonstrate how non-Markovianity can be witnessed directly in the laboratory 
in therms of the expectation value of a hermitian operator}.
These results open huge prospects, allowing the study of non-Markovian dynamics for arbitrary dimensional systems as theoretical and experimental points of view.%\\  

\section{ Entanglement - based measure} 
Given a quantum process described by the quantum channel $\Phi_t\in\M{C}(\B{C}_S)$, where $\M{C}(\B{C}_S)$ is the set of complete positive and trace preserving (CPTP) channels on the Hilbert space $\B{C}_S$. The action of the channel over the density matrix of the system $\rho_S\in\M{D}(\B{C}_S)$, where $\M{D}(\B{C}_S)$ is the set of positive semi-definite operators on $\B{C}_S$, results in the time evolution of the system given by:
\begin{equation}\label{map}
\Phi_t(\rho_S) = \sum_i K_i(t)\rho_s K_i^{\dagger}(t) = \rho_S(t),
\end{equation}
where $\{K_i\}_i$ are the Kraus operators and $\sum_i K_i^{\dagger}K_i = \id_S$.
A sufficient criterion to measure the non-Markovianity of the process $\Phi_t$ can be quantified by the positive variation of the entanglement between the system and a purification ancillary system \cite{plenioprl10}:
\begin{equation}\label{nm1}
\M{N}(\Phi_t) = \int_{\frac{dE}{dt}>0} \frac{dE(\rho_{S:A}(t))}{dt} dt,
\end{equation}
where $E$ is an entanglement measure and $\rho_{S:A}(t))\in\M{D}(\B{C}_S\otimes\B{C}_A)$. It is important to emphasize that the original measure proposed by Rivas, Huelga and Plenio \cite{plenioprl10} assumes that the initial state is maximally entangled. So, following the original prescription we define  

\begin{equation}\label{theorem}
\M{N}(\Phi_t) = \int_{\frac{dE}{dt}>0} \frac{dE(\Phi_t\otimes \id [{\phi_{+}}])}{dt} dt,
\end{equation}
where $\phi_{+} = \ketbra{\phi_{+}}{\phi_{+}}$, with $\ket{\phi_{+}} = \frac{1}{\sqrt{d}}\sum_j \ket{j}_S\ket{j}_A\in\B{C}_S\otimes\B{C}_A$, for $d = dim(S) = dim(A)$ and $\{\ket{j}_X\}_{j = 1}^d$ is an orthonormal basis on $\B{C}_X$.

{Since our measure will be based on the knowledge of the quantum process, it is important to emphasize that the density matrix dynamics determines unequelly the quantum process and vice-versa. Indeed, one could be calculated from the other by means of the {\it Choi - Jamiolkowski isomorphism}} 

\begin{equation}\label{choi}
J(\Phi_t) =  \Phi_t\otimes \id [{\phi_{+}}],
\end{equation}  
{where  $J({\Phi_t})\in\M{D}(\B{C}_S\otimes\B{C}_A)$ is the Choi operator. 
As the map is acting locally over the maximally entangled state, the density matrix that characterize the 
dynamics of the composed system $\rho_{S:A}(t)\in\M{D}(\B{C}_S\otimes\B{C}_A)$ is indeed equal to the Choi operator of the map $\Phi_t$:
\begin{equation}
\rho_{S:A}(t)= J(\Phi_t).  
\end{equation}
The isomorphism of Choi - Jamiolkowski guarantee that the Choi operator $J(\Phi_t)$ is unique, therefore 
there be a linear bijection between $J(\Phi_t)$ and the map $\Phi_t$, which implies that: 
\begin{equation}\label{bij}
\rho_{S:A}(t) \longleftrightarrow  \Phi_t,
\end{equation}
where $\rho_{S:A}(t) = \Phi_t\otimes \id [{\phi_{+}}_{S:A}]$ and $\Phi_t(X_S) = \Tr_A (J(\Phi_t) X_S\otimes \id_A)$, 
for any bounded operator $X_S:\B{C}_S\to\B{C}_S$.}

{Following the approach proposed in Ref.\cite{sabrina17}, 
the variation of the spectral decomposition of 
the operator $J_{\Phi}$  can characterize non - Markovian 
interval of times. 
In this work we present how it is possible 
to quantify the non - Markovianity of $\Phi_t$ by means of the 
expectation value of a witness of {\it non - Markovianity}. 
The witness is defined calculating the  
variation of the entanglement of $J({\Phi}_t)\in\M{D}(\B{C}_S\otimes\B{C}_A)$, 
based on the entanglement witness approach, introduced in the next section.  }

\section{Witnessed Entanglement Approach}
In the study of non-Markovianity process by means of quantum entanglement, the dimension of the composed system, i.e., physical system plus an ancillary system, is defined by the square of the dimension of the former. In this sense, it makes necessary a quantifier of entanglement that can be evaluated even for large dimensional bipartite systems. For this purpose, we use the {\it entanglement witness} approach, which holds for any composed systems \cite{horodeckiew}. Entanglement witnesses are observables that characterize the entanglement of the system \cite{terhal00}. It is defined as a Hermitian operator $W$ acting on $\B{C}_A\otimes \B{C}_B$ such that for an entangled state $\rho\in\M{D}(\B{C}_A\otimes \B{C}_B)$, it holds that:
\begin{eqnarray}
\Tr(W\rho) &<& 0 \\
\Tr(W\sigma) &\geq& 0 \quad\forall\quad \sigma\in Sep.
\end{eqnarray}
where $Sep$ is the set of separable states in $\M{D}(\B{C}_A\otimes\B{C}_B)$. 

One important point about the entanglement witness regards its optimization. A given entanglement witness is named {\it optimal} if it is tangent to the set of separable states \cite{oew} and provides a notion about the distance between the state and the border of the separable set. This notion can be approached to quantify the entanglement of the system. Then it is possible to define a class of entanglement quantifiers named {\it witnessed entanglement} \cite{brandao06,brandaopra2005}, which defines a necessary and sufficient conditions of separability \cite{brprl04}.
A witnessed entanglement quantifier can be written as:
\begin{equation}\label{weq}
E_w(\rho) = \max \left\{0, - \min_{W\in\M{W}} \Tr(W\rho) \right\},
\end{equation}
where $\M{W}$ is the set of entanglement witnesses that distinguishes one measure from another, which means that each set $\M{W}$ results in a different entanglement quantifier. {This result is particularly interesting since several entanglement quantifiers can be written in the witnessed entanglement form,  {once that it defines a class of entanglement monotone quantifiers \cite{brandaopra2005,brandao06}}. In a sense, it generalizes the usual entanglement-based measure of quantum non-Markovianity.} 

An entanglement {monotone} quantifier that can be written in the witnessed entanglement form is the Generalized Robustness of entanglement, that, for a bipartite state $\rho\in \M{D}(\B{C}_A\otimes\B{C}_B)$, is defined as \cite{steiner03}:
\begin{equation}\label{rg}
R_G(\rho) = \min_{\sigma, s\in \Re_{+}}\left\{s: \frac{\rho + s \sigma}{s+1} \in Sep \right\},
\end{equation}
where $\sigma$ is a state that, mixed with $\rho$, results in a separable state. 
Restrictions on the state $\sigma$ in generalized robustness result in different entanglement measures. For example, for $\sigma\in Sep$, it is named \emph{Robustness of Entanglement}, and for $\sigma$ equal to the maximal mixture state, it is named \emph{Random Robustness} \cite{vida99}. Unlike generalized robustness and robustness of entanglement, random robustness is not an entanglement monotone \cite{harrow03}.

Generalized robustness of entanglement can be written as a witnessed entanglement measure \cite{brandaopra2005,brandao06} restricting the set of witness operators by $\M{W} = \{W\in\M{W} | W\leq \id \}$. Therefore:
\begin{equation}\label{rgwn}
R_G(\rho) = \max \left\{0, - \min_{W\leq \id} \Tr(W\rho) \right\}.
\end{equation}
The optimization problem of finding an optimal entanglement witness can be approached as a robust semidefinite program \cite{brprl04,brpra2004}, which implies that the optimization problems in Eq.(\ref{rg}) and Eq.(\ref{rgwn}) attach the same solution, giving equivalent formulation for the same optimization problem \cite{boyd}.%\\

\section{Quantification of non-Markovianity}  
By means of the Generalized Robustness of entanglement, we finally define an evaluable measure of non-Markovianity.
\begin{definition}[Non-Markovianity]
Given the generalized robustness of entanglement $R_G$, we can define a quantifier of non-Markovianity of a process $\Phi_t\in\M{C}(\B{C}_s)$ as
\begin{equation}\label{nm}
\M{N}_{R_G}(\Phi_t) = \int_{\frac{d R_G}{dt}>0} \frac{d R_G(\Phi_t\otimes \id [{\phi_{+}}])}{dt} dt,
\end{equation}
where the integration is taken over all points such that the variation of the entanglement is positive. These intervals of time indicate the non-Markovianity of the process.
\end{definition}

The great advantage of this measure is that both crucial properties of generalized robustness are extended to the quantifier of non-Markovianity in Eq.(\ref{nm}): first, it can be calculated by RSDP for any dimension \cite{brprl04} and, second, it can be written in the entanglement witness form that permits a direct characterization of the non-Markovianity in the laboratory. We can realize the later statement by rewriting the non-Markovianity measure as:
\begin{equation}
\M{N}_{R_G}(\Phi_t) = \int_{\delta R_G>0} \delta R_G(\Phi_t\otimes \id [{\phi_{+}}]) = \sum_i \Delta R_G(i),
\end{equation}
where $\Delta R_G(i) = R_G(t_{i_2}) - R_G(t_{i_1})>0$, for $t_{i_2}>t_{i_1} \,\forall\, i\in\B{N}_1$. For a given entangled state $\rho(t)$, the generalized robustness in witnessed entanglement form is given by:
\begin{equation}
R_G(t) = - \Tr[W_t\rho(t)], 
\end{equation}
where $W_t$ is the optimal entanglement witness for $\rho(t)=\Phi_{t}\otimes\id({\phi_{+}})$. For two intervals of time $t_2>t_1$, the variation of the entanglement holds
\begin{align} \nonumber
\Delta R_G &= R_G(t_2)-R_G(t_1)\\ \nonumber
&= - \left\{ \Tr[W_{t_2}\rho(t_2)] -\Tr[W_{t_1}\rho(t_1)] \right\}\\ \nonumber
& = - \left\{ \Tr[W_{t_2}(\Phi_{t_2}\otimes\id)({\phi_{+}})] -\Tr[W_{t_1}(\Phi_{t_1}\otimes\id)({\phi_{+}})] \right\}\\ \nonumber
& = - \left\{ \Tr([(\Phi^{\dagger}_{t_2}\otimes\id)(W_{t_2})-(\Phi^{\dagger}_{t_1}\otimes\id)(W_{t_1})]\phi_{+})\right\}\\
& = -  \Tr(W(t_2,t_1)\phi_{+}) \label{witnmar}
\end{align}
where $\phi_{+}=\ketbra{\phi_{+}}{\phi_{+}}$, and $W(t_2,t_1) = [(\Phi^{\dagger}_{t_2}\otimes\id)(W_{t_2})-(\Phi^{\dagger}_{t_1}\otimes\id)(W_{t_1})]$ is Hermitian, once $W=W^{\dagger}$. 
Therefore, the variation of the entanglement in each non-Markovian interval $(t_1,t_2)$ can be obtained in the laboratory, measuring the expectation value of the observable $W(t_2,t_1)$ in the maximally entangled state. 

{It is important to emphasize that  Eq.\eqref{witnmar} requires knowledge of the dissipative process, and its characterization usually demand a quantum process tomography. Although by bijection described in Eq.\eqref{bij}, 
the knowledge of the state in each instant of time $t_1,t_2$ guarantees the construction of the map, and 
applicability of Eq.\eqref{witnmar}. This statement can be well understood by means of the representation of 
a entanglement witness $W$ as a positive channel $\Lambda$ \cite{witnesshorodecki}:  
\begin{equation}\label{witmap}
W = \id\otimes\Lambda(\phi_{+}),
\end{equation}   
where $\Lambda$ is positive and trace preserving channel. Therefore, considering $W_{t}$ the optimal 
entanglement witness of a state $\rho(t) = \Phi_{t}\otimes\id({\phi_{+}})$, its expectation value is:
\begin{align}\label{a}
\Tr\left\{  W_{t}\rho(t) \right\} & = \Tr\left\{  W_{t}(\Phi_{t}\otimes\id({\phi_{+}})) \right\} \\\label{b}
& = \Tr\left\{ \phi_{+}(\Phi_{t}\otimes\Lambda({\phi_{+}})) \right\} \\\label{c}
& = \bra{\phi_{+}} \id\otimes\Lambda(\rho_t)\ket{\phi_{+}}, 
\end{align} 
where in Eq.\eqref{b} we used the cyclic property  of the trace and the hermiticity of the entanglement witness.
This result implies that the observable $W(t_2,t_1)$, defined in Eq.\eqref{witnmar}, can be written just in terms
of the states of the system in each instant of time $\{t_1,t_2\}$:
\begin{equation}
W(t_2,t_1) =\id\otimes\Lambda_{t_2}(\rho(t_2)) -  \id\otimes\Lambda_{t_1}(\rho(t_1)),  
\end{equation} 
where $\Lambda_{t_2}$ and $\Lambda_{t_1}$ are the positive operators that characterize the generalized 
robustness in the instants $t_2$ and $t_1$ respectively. 
Therefore, knowing the observable $W(t_2,t_1)$, it is possible to directly measure the diference on the entanglement, in a defined interval of time, characterizing the process about its non-Markovianity in the laboratory. 
It is important to emphasize that the physical implementation of those maps can be performed based on some manipulations and approximations, allowing direct detection of quantum entanglement \cite{PhysRevLett.89.127902, PhysRevA.77.042113, PhysRevLett.107.160401}.
Eq.\eqref{witnmar} creates a new fashion of the characterization and quantification of non-Markovianity as theoretically as experimentally. }

For a given interval of time $(t_1,t_2)$,  where $t_1\leq t_2$, we introduce the {\it witnessed non - Markovianity } : 
\begin{equation}
\M{N}(t_1,t_2) = \max\left\{0, - \bra{\phi_{+}}W(t_2,t_1)\ket{\phi_{+}}\right\}, 
\end{equation}
where $W(t_2,t_1) = [(\Phi^{\dagger}_{t_2}\otimes\id)(W_{t_2})-(\Phi^{\dagger}_{t_1}\otimes\id)(W_{t_1})]$. As $\M{N}(t_1,t_2)$ is not 
zero only for $R_G(t_2)-R_G(t_1)>0$, it is faithful in a given entanglement increasing interval of time. Indeed for a function of 
a quantum process to be a {\it bona fide measure} of non-Markovianity it must satisfy some properties \cite{wolf08,cirac08}:
\begin{enumerate}
\item Faithfulness; 
\item Computability; 
\item Operational or physical interpretation.
\item Invariant under unitary;  
\item \label{prop} Continuity; 
\end{enumerate} 
$\M{N}(t_1,t_2)$ satisfies properties (1) - (3) by its definition based on the entanglement witness of generalized robustness. Property (4) 
comes from the fact that the map is acting locally over the state of the composed system and environment, as $\Phi_{t}\otimes\id({\phi_{+}})$, 
and generalized robustness is invariant under local unitary operations. The continuity is proved by the following theorem: 
\begin{theorem}\label{t1}
Consider two CPTP channel depending on time $\Phi_{t},\Lambda_{t}\in\M{C}(\B{C}_S)$, for $t\in[0,\infty)$, the following inequality holds: 
\begin{equation*}
|\M{N}_{\Phi}(t_1,t_2) - \M{N}_{\Lambda}(t_1,t_2)| \leq d^2 \left( \N{\Phi_{t_1} - \Lambda_{t_1}}_{\diamond}
 + \N{\Phi_{t_2} -\Lambda_{t_2}}_{\diamond} \right),
\end{equation*}
where $d^2=dim(\B{C}_S)dim(\B{C}_A)$, $ \N{\Phi_{t} - \Lambda_{t}}_{\diamond}$ is the diamond distance of $\Phi_{t}$ and $\Lambda_{t}$, that is a continuous function. 
\end{theorem}
The proof of the theorem is presented in the Appendix. This property has an important role in the dynamics of physical systems, once 
that the hamiltonian of the open system depends on some physical parameters, and then variations on those parameters result in different quantum process.

Now, with the defined measure, we are able to study non-Markovian process for arbitrary systems. In the following section we present the dynamics of two trapped cold atoms, a four level system, showing that in this system there exists a crossover between Markov and non-Makovian process in function of the scattering length $a_E$. 

\section{Impurity Atoms Coupled to a Bose-Einstein Condensate} 

\begin{figure}[t]
\begin{center}
\includegraphics[width=.45\textwidth]{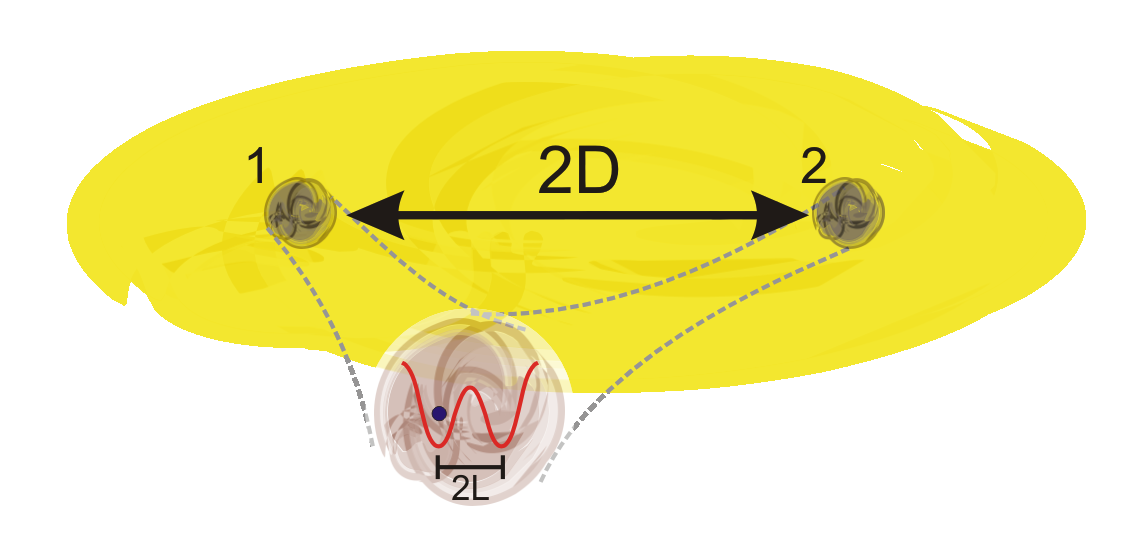}
\end{center}
\caption{(Color Online) An illustrative scheme of the BEC.
The physical system is described by 
two trapped cold atoms (impurities). 
Each atom describes a two-level system trapped in a double well potential, separated by a distance $D$, with an optical superlattice of wavelength $\lambda=4L$. }\label{bec}
\end{figure}

As an example, we study the decoherent dynamics of two trapped cold atoms (impurities), which interact with an ultracold bosonic rubidium gas in a Bose-Einstein condensate (BEC). Each atom is trapped in a double well potential describing, individually, a two-level system in which each base state is associated with the right or left well.
This arrangement has two important characteristics, which can be engineered in the laboratory: the BEC scattering length $a_E$, which is intimately connected with the boson-boson coupling constant to the BEC environment, and the distance $D$ between the impurities. The former modifies the spectral density as a whole and, consequently, it ohmicity. The latter, on the other hand, infers a change similar to that imposed by the spectral density cutoff frequency \cite{sabrina11,sabrina13}.

%\section{Master Equation}
The dynamics of the two impurities coupled to an ultracold bosonic rubidium gas in a BEC is given by a Lindblad-type master equation with
time-dependent decay rates \cite{Cirone,sabrina13}%
\begin{widetext}
\begin{eqnarray}
\frac{d\rho }{dt} &=&\frac{\gamma _{1}\left( t\right) -\gamma _{2}\left(
t\right) }{2}\left[ \left( \sigma _{z}^{(1)}-\sigma _{z}^{(2)}\right) \rho
\left( \sigma _{z}^{(1)}-\sigma _{z}^{(2)}\right) -\frac{1}{2}\left\{ \left(
\sigma _{z}^{(1)}-\sigma _{z}^{(2)}\right) \left( \sigma _{z}^{(1)}-\sigma
_{z}^{(2)}\right) ,\rho \right\} \right]\nonumber  \\
&&+\frac{\gamma _{1}\left( t\right) +\gamma _{2}\left( t\right) }{2}\left[
\left( \sigma _{z}^{(1)}+\sigma _{z}^{(2)}\right) \rho \left( \sigma
_{z}^{(1)}+\sigma _{z}^{(2)}\right) -\frac{1}{2}\left\{ \left( \sigma
_{z}^{(1)}+\sigma _{z}^{(2)}\right) \left( \sigma _{z}^{(1)}+\sigma
_{z}^{(2)}\right) ,\rho \right\} \right] ,\label{becdyna}
\end{eqnarray}\end{widetext}
where $\sigma _{z}^{(n)}$ is the usual Pauli matrix for the $n$-th atom $\left(
n=1,2\right)$, and%
\begin{widetext}
\begin{eqnarray}
\gamma _{1}\left( t\right)  &=&\frac{g_{SE}^{2}n_{0}}{\hbar \pi ^{2}}%
\int_{0}^{\infty }dkk^{2}e^{-k^{2}\sigma ^{2}/2}\frac{\sin \left( \frac{E_{k}%
}{2\hbar }t\right) \cos \left( \frac{E_{k}}{2\hbar }t\right) }{\left(
\epsilon _{k}+2g_{E}n_{0}\right) }\left( 1-\frac{\sin \left( 2kL\right) }{2kL%
}\right)  \\
\gamma _{2}\left( t\right)  &=&\frac{g_{SE}^{2}n_{0}}{2\hbar \pi ^{2}}%
\int_{0}^{\infty }dkk^{2}e^{-k^{2}\sigma ^{2}/2}\frac{\sin \left( \frac{E_{k}%
}{2\hbar }t\right) \cos \left( \frac{E_{k}}{2\hbar }t\right) }{\left(
\epsilon _{k}+2g_{E}n_{0}\right) }\left( \frac{\sin \left( 2k\left(
D+L\right) \right) }{2k\left( D+L\right) }+\frac{\sin \left( 2k\left(
D-L\right) \right) }{2k\left( D-L\right) }-2\frac{\sin \left( 2kD\right) }{%
2kD}\right) ,\label{gam2}
\end{eqnarray}%
\end{widetext}
where the boson-boson coupling of the BEC environment with scattering length $a_{E}$ and the atomic mass $m_{E}$ is given by $g_{E}=4\pi \hbar ^{2}a_{E}/m_{E}$, and $g_{SE}=2\pi \hbar ^{2}a_{SE}/m_{SE}$ is the coupling between the system and
the environment with scattering length $a_{SE}$ and reduced mass $m_{SE}=m_{S}m_{E}/\left( m_{S}+m_{E}\right) $. Also, the energy of the $k$-th Bogoliubov mode is given by $E_{k}=\sqrt{2\epsilon_{k}n_{0}g_{E}+\epsilon 
_{k}^{2}}$, where $n_{0}$ is the condensate density, $\epsilon _{k}=\hbar
^{2}k^{2}/2m_{E}$ and $\sigma $ is the variance parameter of the lattice site. Here  we consider $^{23}\rm{Na}$ impurity atoms immersed in a $^{87}\rm{Rb}$ condensate with $\lambda =600$ nm and $n_{0}=10^{20}$ m$^{-3}$.
The scattering length of the atoms is $a_{Rb}=99a_{0}$ with the Bohr radius given by $a_{0}$, and $a_{SE}=55a_{0}$. 
Finally, the two engineered variables are the distance between the atoms which is given by $2D\geq 8L$, for 
$L=\lambda/4$, and the scattering length of the BEC environment which is given by $a_E$.

It is important to emphasize that to solve this problem exactly, using either the Breuer, Laine, and Pillo measure (based on trace distance) \cite{breuer09} or the Luo, Fu, and Song measure (based on mutual information dynamics) \cite{luo12} is a very difficult task, given the optimization problem involved. Also, even for entanglement based-measures, it is a difficult task since the state of the system plus ancilla is given by a $16\times16$ density matrix. Our measure, on the other hand, overcomes this difficulty, given the power of the RSDP. In Fig.(\ref{fig1}.a), we plot non-Markovianity $\M{N}_{R_G}$ of the BEC dynamics as a function of the scattering length $a_E$ for a separation $D=4L$ (common environment). We obtain a crossover between Markovian and non-Markovian processes for $a_E\approx 0.045 a_{Rb}$, similar to the result presented in Ref. \cite{sabrina11,sabrina13}.

In Fig.(\ref{fig1}.b), on the other hand, we plot non-Markovianity $\M{N}_{R_G}$ of the BEC dynamics as a function of the boson-boson distance $D$. As we can see, the degree of non-Markovianity does not change with the variation in the boson-boson distance, which diverges from the predictions based on other non-Markovianity measures. Indeed, using the BLP measure, C. Addis et al. \cite{sabrina13} show that small distances have a strong effect on the presence of information backflow, tending to a constant value for larger distances. Note, however, that our result is not pointless, once that for decoherent systems, subjected to dephasing errors (as is our example), the map divisibility does not depend on the cutoff frequency, which is intimately connected with the boson-boson distance, which depends on the spectral density shape.\\

\begin{figure}[t]
\begin{center}
\includegraphics[width=.47\textwidth]{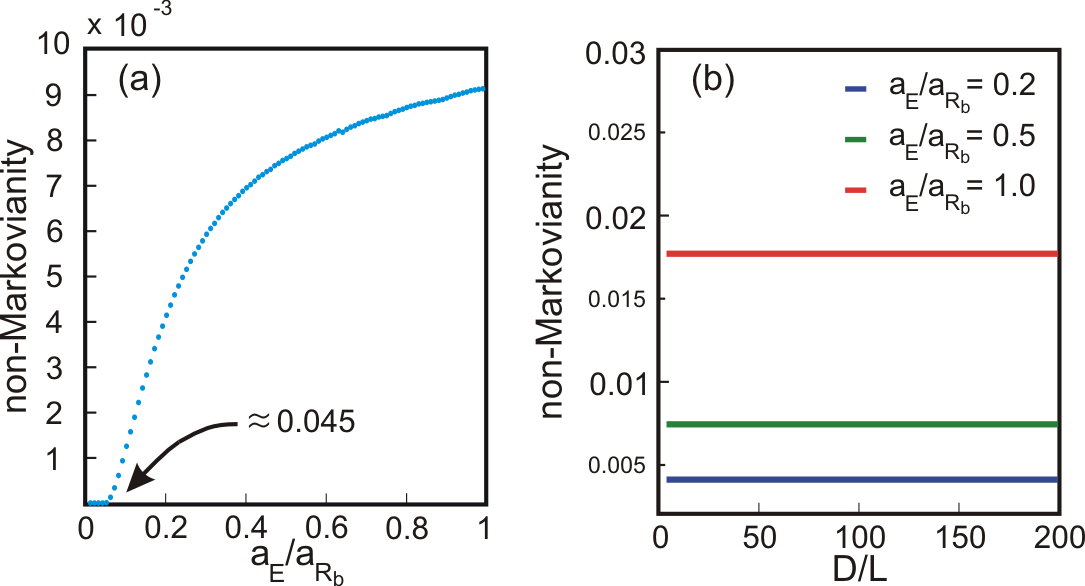} 
\end{center}
\caption{(Color Online) (a) Non-Markovianity $\M{N}_{R_G}$ of BEC dynamics in function of scattering length $a_E/a_{RB}$, for the distance $D = 4L$ (common environment). (b) Non-Markovianity $\M{N}_{R_G}$ of BEC dynamics as a function of the distance $D/L$ for scattering lengths $a_E = 0,2 a_{Rb}$ (blue line); $a_E = 0,5 a_{Rb}$ (green line); and $a_E = 1,0 a_{Rb}$ (red line). For each point, the dynamics for $100$ different instants of time have been calculated, computing the non-Markovianity in Eq.(\ref{nm}). The generalized robustness for each instant of time has been calculated by RSDP. }\label{fig1}
\end{figure}

\section{Conclusion}
%\red{We have presented an analytical solution to the maximization problem involved in the calculation of the entanglement based measure of quantum non-Markovianity. For an arbitrary quantum process, we have demonstrated that the optimal initial bipartite state of the open system and the ancilla is given by a bipartite maximal entangled state. Based on this finding, 

We introduce a computable measure of quantum non-Markovianity based on the witnessed entanglement approach. We have shown how a non-Markovian process can be detected by means of an expectation value of an observable, making it attractive for experimental investigations. {In a sense, the witnessed entanglement approach generalize the usual entanglement based measure of quantum non-Markovianity, since several entanglement quantifiers can be written in an unique formalism. Furthermore, our measure is defined in terms of the quantum process which brings advantages when compared with the density matrix dynamics. 
Indeed, the proposed method can be efficient since it allows the characterization the Markovianity 
of a quantum process, in two intervals of time, by the direct measure of the expectation value of a hermitian operator. 
Besides the operational and experimental advantages of the proposed approach, we also introduce a {\it bona fide} 
measure of non - Markovianity. We also 
proved, in Theorem \ref{t1}, the continuity of the measure, which enable the study of the crossover between Markovian 
and non-Markovian processes. 
Then, we illustrate our results calculating the degree of non-Markovianity of two trapped cold atoms interacting with an 
ultracold bosonic rubidium gas in a Bose-Einstein condensate (BEC). This finding opens up new avenues in the 
investigation of quantum non-Markovian dynamics, providing an efficient way to study dissipative processes for arbitrary dimensional systems. 

\section*{Acknowledgments }
%\acknowledgments 
This work has been supported by the Brazilian agencies FAPESP under the grant number 2015/05581-7, by CNPq under the grant number 474592/2013-8, and by INCT-IQ  under the process number 2008/57856-6. The authors also thank M. Terra Cunha,  S. Maniscalco, R. Vianna and G. Karpat for fruitful discussions. 

\section*{Appendix -  Proof of Property (5)}
To prove the continuity  of $\M{N}(t_1,t_2)$ we first introduce the {\it diamond norm}, that characterizes the notion of distance (distinguishability) of 
quantum maps. 
\begin{definition}[Diamond Norm]\label{diamond}
Given a CPTP channel $\Phi_S\in \M{C}(\B{C}_S)$, its diamond norm is defined as: 
\begin{equation}
\N{\Phi_S}_{\diamond} = \max_{\ket{\psi}\in\B{C}_S\otimes\B{C}_A} \N{ \Phi_S\otimes\id_A(\psi)}_1, 
\end{equation}
where $\psi = \ketbra{\psi}{\psi}$.
\end{definition}

The continuity is proved by of the following theorem: 
\begin{theoremn}
Consider two CPTP channel depending on time $\Phi_{t},\Lambda_{t}\in\M{C}(\B{C}_S)$, for $t\in[0,\infty)$, the following inequality holds: 
\begin{equation*}
|\M{N}_{\Phi}(t_1,t_2) - \M{N}_{\Lambda}(t_1,t_2)| \leq d \left( \N{\Phi_{t_1} - \Lambda_{t_1}}_{\diamond}
 + \N{\Phi_{t_2} -\Lambda_{t_2}}_{\diamond} \right),
\end{equation*}
where $ \N{\Phi_{t} - \Lambda_{t}}_{\diamond}$ is the diamond distance of $\Phi_{t}$ and $\Lambda_{t}$, that is a continuous function. 
\end{theoremn}
\begin{proof}
Given the definition of witnessed non - Markovianity $\M{N}(t_1,t_2) = \max\{0, - \bra{\phi_{+}}(W(t_2,t_1)\ket{\phi_{+}}\}$, then:
\begin{align}
|\M{N}_{\Phi}(t_1,t_2) - \M{N}_{\Lambda}(t_1,t_2)| =& |\bra{\phi_{+}}(W(t_2,t_1)\ket{\phi_{+}}\\ &- \bra{\phi_{+}}(M(t_2,t_1)\ket{\phi_{+}}|, 
\end{align} 
where $W(t_2,t_1)= [(\Phi^{\dagger}_{t_2}\otimes\id)(W_{t_2})-(\Phi^{\dagger}_{t_1}\otimes\id)(W_{t_1})]$ and 
$M(t_2,t_1)= [(\Lambda^{\dagger}_{t_2}\otimes\id)(M_{t_2})-(\Lambda^{\dagger}_{t_1}\otimes\id)(M_{t_1})]$. Therefore applying the triangle inequality:  
\begin{align}
&|\M{N}_{\Phi}(t_1,t_2) - \M{N}_{\Lambda}(t_1,t_2)| \\
\leq &\left| \Tr\left\{[ \Phi^{\dagger}_{t_1}\otimes\id(W_{t_1}) -  \Lambda^{\dagger}_{t_1}\otimes\id(M_{t_1})]\phi_{+} \right\}  \right| \\
 &+ \left| \Tr\left\{ [\Phi^{\dagger}_{t_2}\otimes\id(W_{t_2}) -  \Lambda^{\dagger}_{t_2}\otimes\id(M_{t_2}) ]\phi_{+}\right\}  \right|.  
\end{align} 
Although each term of the sum respects: 
\begin{align}
&\left| \Tr\left\{[ \Phi^{\dagger}_{t}\otimes\id(W_{t}) -  \Lambda^{\dagger}_{t}\otimes\id(M_{t})]\phi_{+} \right\}  \right|  \\
=& \left| \Tr\left\{[ W_{t} \Phi_{t}\otimes\id(\phi_{+}) - M_{t} \Lambda_{t}\otimes\id(\phi_{+}) \right\}  \right|\\ \label{wm}
\leq & \left| \Tr\left\{[ W_{t} \Phi_{t}\otimes\id(\phi_{+}) - W_{t} \Lambda_{t}\otimes\id(\phi_{+}) \right\}  \right|\\ \label{holder} 
\leq & \N{W_{t}}_{1} \N{\Phi_{t}\otimes\id(\phi_{+}) - \Lambda_{t}\otimes\id(\phi_{+})}_1 \\ \label{defi} 
\leq & \N{W_{t}}_{1} \N{\Phi_{t}\otimes\id(\phi_{+}) - \Lambda_{t}\otimes\id(\phi_{+})}_{\diamond}.
\end{align}
Where in Eq.\eqref{wm} we have used that $M_t$ is the optimal entanglement witness for $\Lambda_{t}\otimes\id(\phi_{+})$, then 
$\Tr(M_{t} \Lambda_{t}\otimes\id(\phi_{+}))\leq \Tr(W_{t} \Lambda_{t}\otimes\id(\phi_{+})) $; in Eq.\eqref{holder} we have used the Holder inequality: 
$\Tr(AB)\leq \N{A}_p\N{B}_q$ for $\frac{1}{q}+\frac{1}{p} =1$ and $A$ and $B$ hermitian; in Eq.\eqref{defi} we have apply the definition of 
diamond norm. As the entanglement based measure of non - Markonianity is defined by means of generalized robustness of entanglement, 
the optimal entanglement witness satisfies $W\leq \id$ \cite{brandaopra2005}, then $\N{\id}_1=d^2$, which proves the theorem.  
\end{proof}

%
%\bibliography{rev4pratdd.bib} 

\begin{thebibliography}{31}%
\makeatletter
\providecommand \@ifxundefined [1]{%
 \@ifx{#1\undefined}
}%
\providecommand \@ifnum [1]{%
 \ifnum #1\expandafter \@firstoftwo
 \else \expandafter \@secondoftwo
 \fi
}%
\providecommand \@ifx [1]{%
 \ifx #1\expandafter \@firstoftwo
 \else \expandafter \@secondoftwo
 \fi
}%
\providecommand \natexlab [1]{#1}%
\providecommand \enquote  [1]{``#1''}%
\providecommand \bibnamefont  [1]{#1}%
\providecommand \bibfnamefont [1]{#1}%
\providecommand \citenamefont [1]{#1}%
\providecommand \href@noop [0]{\@secondoftwo}%
\providecommand \href [0]{\begingroup \@sanitize@url \@href}%
\providecommand \@href[1]{\@@startlink{#1}\@@href}%
\providecommand \@@href[1]{\endgroup#1\@@endlink}%
\providecommand \@sanitize@url [0]{\catcode `\\12\catcode `\$12\catcode
  `\&12\catcode `\#12\catcode `\^12\catcode `\_12\catcode `\%12\relax}%
\providecommand \@@startlink[1]{}%
\providecommand \@@endlink[0]{}%
\providecommand \url  [0]{\begingroup\@sanitize@url \@url }%
\providecommand \@url [1]{\endgroup\@href {#1}{\urlprefix }}%
\providecommand \urlprefix  [0]{URL }%
\providecommand \Eprint [0]{\href }%
\providecommand \doibase [0]{http://dx.doi.org/}%
\providecommand \selectlanguage [0]{\@gobble}%
\providecommand \bibinfo  [0]{\@secondoftwo}%
\providecommand \bibfield  [0]{\@secondoftwo}%
\providecommand \translation [1]{[#1]}%
\providecommand \BibitemOpen [0]{}%
\providecommand \bibitemStop [0]{}%
\providecommand \bibitemNoStop [0]{.\EOS\space}%
\providecommand \EOS [0]{\spacefactor3000\relax}%
\providecommand \BibitemShut  [1]{\csname bibitem#1\endcsname}%
\let\auto@bib@innerbib\@empty
%</preamble>
\bibitem [{\citenamefont {Breuer}\ \emph {et~al.}(2016)\citenamefont {Breuer},
  \citenamefont {Laine}, \citenamefont {Piilo},\ and\ \citenamefont
  {Vacchini}}]{breuer16}%
  \BibitemOpen
  \bibfield  {author} {\bibinfo {author} {\bibfnamefont {H.-P.}\ \bibnamefont
  {Breuer}}, \bibinfo {author} {\bibfnamefont {E.-M.}\ \bibnamefont {Laine}},
  \bibinfo {author} {\bibfnamefont {J.}~\bibnamefont {Piilo}}, \ and\ \bibinfo
  {author} {\bibfnamefont {B.}~\bibnamefont {Vacchini}},\ }\href {\doibase
  10.1103/RevModPhys.88.021002} {\bibfield  {journal} {\bibinfo  {journal}
  {Rev. Mod. Phys.}\ }\textbf {\bibinfo {volume} {88}},\ \bibinfo {pages}
  {021002} (\bibinfo {year} {2016})}\BibitemShut {NoStop}%
\bibitem [{\citenamefont {Nielsen}\ and\ \citenamefont
  {Chuang}(2000)}]{chuang}%
  \BibitemOpen
  \bibfield  {author} {\bibinfo {author} {\bibfnamefont {M.}~\bibnamefont
  {Nielsen}}\ and\ \bibinfo {author} {\bibfnamefont {I.}~\bibnamefont
  {Chuang}},\ }\href {http://books.google.com.br/books?id=65FqEKQOfP8C} {\emph
  {\bibinfo {title} {{Quantum Computation and Quantum Information}}}},\
  {Cambridge Series on Information and the Natural Sciences}\ (\bibinfo
  {publisher} {Cambridge University Press},\ \bibinfo {year}
  {2000})\BibitemShut {NoStop}%
\bibitem [{\citenamefont {Breuer}\ \emph {et~al.}(2009)\citenamefont {Breuer},
  \citenamefont {Laine},\ and\ \citenamefont {Piilo}}]{breuer09}%
  \BibitemOpen
  \bibfield  {author} {\bibinfo {author} {\bibfnamefont {H.-P.}\ \bibnamefont
  {Breuer}}, \bibinfo {author} {\bibfnamefont {E.-M.}\ \bibnamefont {Laine}}, \
  and\ \bibinfo {author} {\bibfnamefont {J.}~\bibnamefont {Piilo}},\
  }\href@noop {} {\bibfield  {journal} {\bibinfo  {journal} {Physical review
  letters}\ }\textbf {\bibinfo {volume} {103}},\ \bibinfo {pages} {210401}
  (\bibinfo {year} {2009})}\BibitemShut {NoStop}%
\bibitem [{\citenamefont {Rivas}\ \emph
  {et~al.}(2010{\natexlab{a}})\citenamefont {Rivas}, \citenamefont {Huelga},\
  and\ \citenamefont {Plenio}}]{plenioprl10}%
  \BibitemOpen
  \bibfield  {author} {\bibinfo {author} {\bibfnamefont {{\'A}.}~\bibnamefont
  {Rivas}}, \bibinfo {author} {\bibfnamefont {S.~F.}\ \bibnamefont {Huelga}}, \
  and\ \bibinfo {author} {\bibfnamefont {M.~B.}\ \bibnamefont {Plenio}},\
  }\href {\doibase 10.1103/PhysRevLett.105.050403} {\bibfield  {journal}
  {\bibinfo  {journal} {Phys. Rev. Lett.}\ }\textbf {\bibinfo {volume} {105}},\
  \bibinfo {pages} {050403} (\bibinfo {year} {2010}{\natexlab{a}})}\BibitemShut
  {NoStop}%
\bibitem [{\citenamefont {Luo}\ \emph {et~al.}(2012)\citenamefont {Luo},
  \citenamefont {Fu},\ and\ \citenamefont {Song}}]{luo12}%
  \BibitemOpen
  \bibfield  {author} {\bibinfo {author} {\bibfnamefont {S.}~\bibnamefont
  {Luo}}, \bibinfo {author} {\bibfnamefont {S.}~\bibnamefont {Fu}}, \ and\
  \bibinfo {author} {\bibfnamefont {H.}~\bibnamefont {Song}},\ }\href@noop {}
  {\bibfield  {journal} {\bibinfo  {journal} {Physical Review A}\ }\textbf
  {\bibinfo {volume} {86}},\ \bibinfo {pages} {044101} (\bibinfo {year}
  {2012})}\BibitemShut {NoStop}%
\bibitem [{\citenamefont {Rivas}\ \emph
  {et~al.}(2010{\natexlab{b}})\citenamefont {Rivas}, \citenamefont {Huelga},\
  and\ \citenamefont {Plenio}}]{nm1}%
  \BibitemOpen
  \bibfield  {author} {\bibinfo {author} {\bibfnamefont {A.}~\bibnamefont
  {Rivas}}, \bibinfo {author} {\bibfnamefont {S.~F.}\ \bibnamefont {Huelga}}, \
  and\ \bibinfo {author} {\bibfnamefont {M.~B.}\ \bibnamefont {Plenio}},\
  }\href@noop {} {\bibfield  {journal} {\bibinfo  {journal} {Physical Review
  Letters}\ }\textbf {\bibinfo {volume} {105}},\ \bibinfo {pages} {05403}
  (\bibinfo {year} {2010}{\natexlab{b}})}\BibitemShut {NoStop}%
\bibitem [{\citenamefont {Laine}\ \emph {et~al.}(2010)\citenamefont {Laine},
  \citenamefont {Piilo},\ and\ \citenamefont {Breuer}}]{nm2}%
  \BibitemOpen
  \bibfield  {author} {\bibinfo {author} {\bibfnamefont {E.-M.}\ \bibnamefont
  {Laine}}, \bibinfo {author} {\bibfnamefont {J.}~\bibnamefont {Piilo}}, \ and\
  \bibinfo {author} {\bibfnamefont {H.-P.}\ \bibnamefont {Breuer}},\
  }\href@noop {} {\bibfield  {journal} {\bibinfo  {journal} {Physical Review
  A}\ }\textbf {\bibinfo {volume} {81}},\ \bibinfo {pages} {062115} (\bibinfo
  {year} {2010})}\BibitemShut {NoStop}%
\bibitem [{\citenamefont {Addis}\ \emph {et~al.}(2014)\citenamefont {Addis},
  \citenamefont {Bylicka}, \citenamefont {Chru{\'s}ci{\'n}ski},\ and\
  \citenamefont {Maniscalco}}]{nm3}%
  \BibitemOpen
  \bibfield  {author} {\bibinfo {author} {\bibfnamefont {C.}~\bibnamefont
  {Addis}}, \bibinfo {author} {\bibfnamefont {B.}~\bibnamefont {Bylicka}},
  \bibinfo {author} {\bibfnamefont {D.}~\bibnamefont {Chru{\'s}ci{\'n}ski}}, \
  and\ \bibinfo {author} {\bibfnamefont {S.}~\bibnamefont {Maniscalco}},\
  }\href@noop {} {\bibfield  {journal} {\bibinfo  {journal} {Physical Review
  A}\ }\textbf {\bibinfo {volume} {90}},\ \bibinfo {pages} {052103} (\bibinfo
  {year} {2014})}\BibitemShut {NoStop}%
\bibitem [{\citenamefont {Neto}\ \emph {et~al.}(2016)\citenamefont {Neto},
  \citenamefont {Karpat},\ and\ \citenamefont {Fanchini}}]{neto16}%
  \BibitemOpen
  \bibfield  {author} {\bibinfo {author} {\bibfnamefont {A.~C.}\ \bibnamefont
  {Neto}}, \bibinfo {author} {\bibfnamefont {G.}~\bibnamefont {Karpat}}, \ and\
  \bibinfo {author} {\bibfnamefont {F.~F.}\ \bibnamefont {Fanchini}},\
  }\href@noop {} {\bibfield  {journal} {\bibinfo  {journal} {Physical Review
  A}\ }\textbf {\bibinfo {volume} {94}},\ \bibinfo {pages} {032105} (\bibinfo
  {year} {2016})}\BibitemShut {NoStop}%
\bibitem [{\citenamefont {Brand{\~a}o}(2005)}]{brandaopra2005}%
  \BibitemOpen
  \bibfield  {author} {\bibinfo {author} {\bibfnamefont {F.~G. S.~L.}\
  \bibnamefont {Brand{\~a}o}},\ }\href {\doibase 10.1103/PhysRevA.72.022310}
  {\bibfield  {journal} {\bibinfo  {journal} {Phys. Rev. A}\ }\textbf {\bibinfo
  {volume} {72}},\ \bibinfo {pages} {022310} (\bibinfo {year}
  {2005})}\BibitemShut {NoStop}%
\bibitem [{\citenamefont {Brandao}\ and\ \citenamefont
  {Vianna}(2004)}]{rsdpreinaldo}%
  \BibitemOpen
  \bibfield  {author} {\bibinfo {author} {\bibfnamefont {F.~G. S.~L.}\
  \bibnamefont {Brandao}}\ and\ \bibinfo {author} {\bibfnamefont {R.~O.}\
  \bibnamefont {Vianna}},\ }\href {\doibase 10.1103/PhysRevA.70.062309}
  {\bibfield  {journal} {\bibinfo  {journal} {Phys. Rev. A}\ }\textbf {\bibinfo
  {volume} {70}},\ \bibinfo {pages} {062309} (\bibinfo {year} {2004})},\
  \Eprint {http://arxiv.org/abs/quant-ph/0405008v2} {arXiv:quant-ph/0405008v2
  [quant-ph]} \BibitemShut {NoStop}%
\bibitem [{\citenamefont {Chru{\'s}ci{\'n}ski}\ \emph
  {et~al.}(2017)\citenamefont {Chru{\'s}ci{\'n}ski}, \citenamefont
  {Macchiavello},\ and\ \citenamefont {Maniscalco}}]{sabrina17}%
  \BibitemOpen
  \bibfield  {author} {\bibinfo {author} {\bibfnamefont {D.}~\bibnamefont
  {Chru{\'s}ci{\'n}ski}}, \bibinfo {author} {\bibfnamefont {C.}~\bibnamefont
  {Macchiavello}}, \ and\ \bibinfo {author} {\bibfnamefont {S.}~\bibnamefont
  {Maniscalco}},\ }\href {\doibase 10.1103/PhysRevLett.118.080404} {\bibfield
  {journal} {\bibinfo  {journal} {Phys. Rev. Lett.}\ }\textbf {\bibinfo
  {volume} {118}},\ \bibinfo {pages} {080404} (\bibinfo {year}
  {2017})}\BibitemShut {NoStop}%
\bibitem [{\citenamefont {Horodecki}\ \emph
  {et~al.}(1996{\natexlab{a}})\citenamefont {Horodecki}, \citenamefont
  {Horodecki},\ and\ \citenamefont {Horodecki}}]{horodeckiew}%
  \BibitemOpen
  \bibfield  {author} {\bibinfo {author} {\bibfnamefont {M.}~\bibnamefont
  {Horodecki}}, \bibinfo {author} {\bibfnamefont {P.}~\bibnamefont
  {Horodecki}}, \ and\ \bibinfo {author} {\bibfnamefont {R.}~\bibnamefont
  {Horodecki}},\ }\href@noop {} {\bibfield  {journal} {\bibinfo  {journal}
  {Physics Letters A}\ }\textbf {\bibinfo {volume} {223}},\ \bibinfo {pages}
  {1} (\bibinfo {year} {1996}{\natexlab{a}})}\BibitemShut {NoStop}%
\bibitem [{\citenamefont {Terhal}(2000)}]{terhal00}%
  \BibitemOpen
  \bibfield  {author} {\bibinfo {author} {\bibfnamefont {B.~M.}\ \bibnamefont
  {Terhal}},\ }\href {\doibase 10.1016/S0375-9601(00)00401-1} {\bibfield
  {journal} {\bibinfo  {journal} {Physics Letters A Vol. 271, 319 (2000)}\ }
  (\bibinfo {year} {2000}),\ 10.1016/S0375-9601(00)00401-1},\ \Eprint
  {http://arxiv.org/abs/quant-ph/9911057v4} {arXiv:quant-ph/9911057v4
  [quant-ph]} \BibitemShut {NoStop}%
\bibitem [{\citenamefont {Lewenstein}\ \emph {et~al.}(2000)\citenamefont
  {Lewenstein}, \citenamefont {Kraus}, \citenamefont {Cirac},\ and\
  \citenamefont {Horodecki}}]{oew}%
  \BibitemOpen
  \bibfield  {author} {\bibinfo {author} {\bibfnamefont {M.}~\bibnamefont
  {Lewenstein}}, \bibinfo {author} {\bibfnamefont {B.}~\bibnamefont {Kraus}},
  \bibinfo {author} {\bibfnamefont {J.~I.}\ \bibnamefont {Cirac}}, \ and\
  \bibinfo {author} {\bibfnamefont {P.}~\bibnamefont {Horodecki}},\ }\href
  {\doibase 10.1103/PhysRevA.62.052310} {\bibfield  {journal} {\bibinfo
  {journal} {Phys. Rev. A}\ }\textbf {\bibinfo {volume} {62}},\ \bibinfo
  {pages} {052310} (\bibinfo {year} {2000})}\BibitemShut {NoStop}%
\bibitem [{\citenamefont {Brand{\~a}o}\ and\ \citenamefont
  {Vianna}(2006)}]{brandao06}%
  \BibitemOpen
  \bibfield  {author} {\bibinfo {author} {\bibfnamefont {F.~G.}\ \bibnamefont
  {Brand{\~a}o}}\ and\ \bibinfo {author} {\bibfnamefont {R.~O.}\ \bibnamefont
  {Vianna}},\ }\href {\doibase 10.1142/S0219749906001803} {\bibfield  {journal}
  {\bibinfo  {journal} {Int. J. of Quant. Inf.}\ }\textbf {\bibinfo {volume}
  {4}},\ \bibinfo {pages} {331} (\bibinfo {year} {2006})}\BibitemShut {NoStop}%
\bibitem [{\citenamefont {Brand{\~a}o}\ and\ \citenamefont
  {Vianna}(2004{\natexlab{a}})}]{brprl04}%
  \BibitemOpen
  \bibfield  {author} {\bibinfo {author} {\bibfnamefont {F.~G. S.~L.}\
  \bibnamefont {Brand{\~a}o}}\ and\ \bibinfo {author} {\bibfnamefont {R.~O.}\
  \bibnamefont {Vianna}},\ }\href {\doibase 10.1103/PhysRevLett.93.220503}
  {\bibfield  {journal} {\bibinfo  {journal} {Phys. Rev. Lett.}\ }\textbf
  {\bibinfo {volume} {93}},\ \bibinfo {pages} {220503} (\bibinfo {year}
  {2004}{\natexlab{a}})}\BibitemShut {NoStop}%
\bibitem [{\citenamefont {Steiner}(2003)}]{steiner03}%
  \BibitemOpen
  \bibfield  {author} {\bibinfo {author} {\bibfnamefont {M.}~\bibnamefont
  {Steiner}},\ }\href@noop {} {\bibfield  {journal} {\bibinfo  {journal}
  {Physical Review A}\ }\textbf {\bibinfo {volume} {67}},\ \bibinfo {pages}
  {054305} (\bibinfo {year} {2003})}\BibitemShut {NoStop}%
\bibitem [{\citenamefont {Vidal}\ and\ \citenamefont {Tarrach}(1999)}]{vida99}%
  \BibitemOpen
  \bibfield  {author} {\bibinfo {author} {\bibfnamefont {G.}~\bibnamefont
  {Vidal}}\ and\ \bibinfo {author} {\bibfnamefont {R.}~\bibnamefont
  {Tarrach}},\ }\href@noop {} {\bibfield  {journal} {\bibinfo  {journal}
  {Physical Review A}\ }\textbf {\bibinfo {volume} {59}},\ \bibinfo {pages}
  {141} (\bibinfo {year} {1999})}\BibitemShut {NoStop}%
\bibitem [{\citenamefont {Harrow}\ and\ \citenamefont
  {Nielsen}(2003)}]{harrow03}%
  \BibitemOpen
  \bibfield  {author} {\bibinfo {author} {\bibfnamefont {A.~W.}\ \bibnamefont
  {Harrow}}\ and\ \bibinfo {author} {\bibfnamefont {M.~A.}\ \bibnamefont
  {Nielsen}},\ }\href@noop {} {\bibfield  {journal} {\bibinfo  {journal}
  {Physical Review A}\ }\textbf {\bibinfo {volume} {68}},\ \bibinfo {pages}
  {012308} (\bibinfo {year} {2003})}\BibitemShut {NoStop}%
\bibitem [{\citenamefont {Brand{\~a}o}\ and\ \citenamefont
  {Vianna}(2004{\natexlab{b}})}]{brpra2004}%
  \BibitemOpen
  \bibfield  {author} {\bibinfo {author} {\bibfnamefont {F.~G. S.~L.}\
  \bibnamefont {Brand{\~a}o}}\ and\ \bibinfo {author} {\bibfnamefont {R.~O.}\
  \bibnamefont {Vianna}},\ }\href {\doibase 10.1103/PhysRevA.70.062309}
  {\bibfield  {journal} {\bibinfo  {journal} {Phys. Rev. A}\ }\textbf {\bibinfo
  {volume} {70}},\ \bibinfo {pages} {062309} (\bibinfo {year}
  {2004}{\natexlab{b}})}\BibitemShut {NoStop}%
\bibitem [{\citenamefont {Boyd}\ and\ \citenamefont
  {Vandenberghe}(2004)}]{boyd}%
  \BibitemOpen
  \bibfield  {author} {\bibinfo {author} {\bibfnamefont {S.}~\bibnamefont
  {Boyd}}\ and\ \bibinfo {author} {\bibfnamefont {L.}~\bibnamefont
  {Vandenberghe}},\ }\href {http://books.google.com.br/books?id=mYm0bLd3fcoC}
  {\emph {\bibinfo {title} {{Convex Optimization}}}},\ {Berichte {\"u}ber
  verteilte messysteme}\ (\bibinfo  {publisher} {Cambridge University Press},\
  \bibinfo {year} {2004})\BibitemShut {NoStop}%
\bibitem [{\citenamefont {Horodecki}\ \emph
  {et~al.}(1996{\natexlab{b}})\citenamefont {Horodecki}, \citenamefont
  {Horodecki},\ and\ \citenamefont {Horodecki}}]{witnesshorodecki}%
  \BibitemOpen
  \bibfield  {author} {\bibinfo {author} {\bibfnamefont {M.}~\bibnamefont
  {Horodecki}}, \bibinfo {author} {\bibfnamefont {P.}~\bibnamefont
  {Horodecki}}, \ and\ \bibinfo {author} {\bibfnamefont {R.}~\bibnamefont
  {Horodecki}},\ }\href {\doibase 10.1016/S0375-9601(96)00706-2} {\bibfield
  {journal} {\bibinfo  {journal} {Physics Letters A}\ }\textbf {\bibinfo
  {volume} {223}},\ \bibinfo {pages} {1} (\bibinfo {year}
  {1996}{\natexlab{b}})}\BibitemShut {NoStop}%
\bibitem [{\citenamefont {Horodecki}\ and\ \citenamefont
  {Ekert}(2002)}]{PhysRevLett.89.127902}%
  \BibitemOpen
  \bibfield  {author} {\bibinfo {author} {\bibfnamefont {P.}~\bibnamefont
  {Horodecki}}\ and\ \bibinfo {author} {\bibfnamefont {A.}~\bibnamefont
  {Ekert}},\ }\href {\doibase 10.1103/PhysRevLett.89.127902} {\bibfield
  {journal} {\bibinfo  {journal} {Phys. Rev. Lett.}\ }\textbf {\bibinfo
  {volume} {89}},\ \bibinfo {pages} {127902} (\bibinfo {year}
  {2002})}\BibitemShut {NoStop}%
\bibitem [{\citenamefont {Carteret}\ \emph {et~al.}(2008)\citenamefont
  {Carteret}, \citenamefont {Terno},\ and\ \citenamefont {\ifmmode~\dot{Z}\else
  \.{Z}\fi{}yczkowski}}]{PhysRevA.77.042113}%
  \BibitemOpen
  \bibfield  {author} {\bibinfo {author} {\bibfnamefont {H.~A.}\ \bibnamefont
  {Carteret}}, \bibinfo {author} {\bibfnamefont {D.~R.}\ \bibnamefont {Terno}},
  \ and\ \bibinfo {author} {\bibfnamefont {K.}~\bibnamefont
  {\ifmmode~\dot{Z}\else \.{Z}\fi{}yczkowski}},\ }\href {\doibase
  10.1103/PhysRevA.77.042113} {\bibfield  {journal} {\bibinfo  {journal} {Phys.
  Rev. A}\ }\textbf {\bibinfo {volume} {77}},\ \bibinfo {pages} {042113}
  (\bibinfo {year} {2008})}\BibitemShut {NoStop}%
\bibitem [{\citenamefont {Lim}\ \emph {et~al.}(2011)\citenamefont {Lim},
  \citenamefont {Kim}, \citenamefont {Ra}, \citenamefont {Bae},\ and\
  \citenamefont {Kim}}]{PhysRevLett.107.160401}%
  \BibitemOpen
  \bibfield  {author} {\bibinfo {author} {\bibfnamefont {H.-T.}\ \bibnamefont
  {Lim}}, \bibinfo {author} {\bibfnamefont {Y.-S.}\ \bibnamefont {Kim}},
  \bibinfo {author} {\bibfnamefont {Y.-S.}\ \bibnamefont {Ra}}, \bibinfo
  {author} {\bibfnamefont {J.}~\bibnamefont {Bae}}, \ and\ \bibinfo {author}
  {\bibfnamefont {Y.-H.}\ \bibnamefont {Kim}},\ }\href {\doibase
  10.1103/PhysRevLett.107.160401} {\bibfield  {journal} {\bibinfo  {journal}
  {Phys. Rev. Lett.}\ }\textbf {\bibinfo {volume} {107}},\ \bibinfo {pages}
  {160401} (\bibinfo {year} {2011})}\BibitemShut {NoStop}%
\bibitem [{\citenamefont {Wolf}\ \emph {et~al.}(2008)\citenamefont {Wolf},
  \citenamefont {Eisert}, \citenamefont {Cubitt},\ and\ \citenamefont
  {Cirac}}]{wolf08}%
  \BibitemOpen
  \bibfield  {author} {\bibinfo {author} {\bibfnamefont {M.~M.}\ \bibnamefont
  {Wolf}}, \bibinfo {author} {\bibfnamefont {J.}~\bibnamefont {Eisert}},
  \bibinfo {author} {\bibfnamefont {T.~S.}\ \bibnamefont {Cubitt}}, \ and\
  \bibinfo {author} {\bibfnamefont {J.~I.}\ \bibnamefont {Cirac}},\ }\href
  {\doibase 10.1103/PhysRevLett.101.150402} {\bibfield  {journal} {\bibinfo
  {journal} {Phys. Rev. Lett.}\ }\textbf {\bibinfo {volume} {101}},\ \bibinfo
  {pages} {150402} (\bibinfo {year} {2008})}\BibitemShut {NoStop}%
\bibitem [{\citenamefont {Wolf}\ and\ \citenamefont {Cirac}(2008)}]{cirac08}%
  \BibitemOpen
  \bibfield  {author} {\bibinfo {author} {\bibfnamefont {M.~M.}\ \bibnamefont
  {Wolf}}\ and\ \bibinfo {author} {\bibfnamefont {J.~I.}\ \bibnamefont
  {Cirac}},\ }\href {\doibase 10.1007/s00220-008-0411-y} {\bibfield  {journal}
  {\bibinfo  {journal} {Communications in Mathematical Physics}\ }\textbf
  {\bibinfo {volume} {279}},\ \bibinfo {pages} {147} (\bibinfo {year}
  {2008})}\BibitemShut {NoStop}%
\bibitem [{\citenamefont {Haikka}\ \emph {et~al.}(2011)\citenamefont {Haikka},
  \citenamefont {McEndoo}, \citenamefont {{De Chiara}}, \citenamefont {Palma},\
  and\ \citenamefont {Maniscalco}}]{sabrina11}%
  \BibitemOpen
  \bibfield  {author} {\bibinfo {author} {\bibfnamefont {P.}~\bibnamefont
  {Haikka}}, \bibinfo {author} {\bibfnamefont {S.}~\bibnamefont {McEndoo}},
  \bibinfo {author} {\bibfnamefont {G.}~\bibnamefont {{De Chiara}}}, \bibinfo
  {author} {\bibfnamefont {G.~M.}\ \bibnamefont {Palma}}, \ and\ \bibinfo
  {author} {\bibfnamefont {S.}~\bibnamefont {Maniscalco}},\ }\href {\doibase
  10.1103/PhysRevA.84.031602} {\bibfield  {journal} {\bibinfo  {journal} {Phys.
  Rev. A}\ }\textbf {\bibinfo {volume} {84}},\ \bibinfo {pages} {031602}
  (\bibinfo {year} {2011})}\BibitemShut {NoStop}%
\bibitem [{\citenamefont {Addis}\ \emph {et~al.}(2013)\citenamefont {Addis},
  \citenamefont {Haikka}, \citenamefont {McEndoo}, \citenamefont
  {Macchiavello},\ and\ \citenamefont {Maniscalco}}]{sabrina13}%
  \BibitemOpen
  \bibfield  {author} {\bibinfo {author} {\bibfnamefont {C.}~\bibnamefont
  {Addis}}, \bibinfo {author} {\bibfnamefont {P.}~\bibnamefont {Haikka}},
  \bibinfo {author} {\bibfnamefont {S.}~\bibnamefont {McEndoo}}, \bibinfo
  {author} {\bibfnamefont {C.}~\bibnamefont {Macchiavello}}, \ and\ \bibinfo
  {author} {\bibfnamefont {S.}~\bibnamefont {Maniscalco}},\ }\href {\doibase
  10.1103/PhysRevA.87.052109} {\bibfield  {journal} {\bibinfo  {journal} {Phys.
  Rev. A}\ }\textbf {\bibinfo {volume} {87}},\ \bibinfo {pages} {052109}
  (\bibinfo {year} {2013})}\BibitemShut {NoStop}%
\bibitem [{\citenamefont {Cirone}\ \emph {et~al.}(2009)\citenamefont {Cirone},
  \citenamefont {Chiara}, \citenamefont {Palma},\ and\ \citenamefont
  {Recati}}]{Cirone}%
  \BibitemOpen
  \bibfield  {author} {\bibinfo {author} {\bibfnamefont {M.~A.}\ \bibnamefont
  {Cirone}}, \bibinfo {author} {\bibfnamefont {G.~D.}\ \bibnamefont {Chiara}},
  \bibinfo {author} {\bibfnamefont {G.~M.}\ \bibnamefont {Palma}}, \ and\
  \bibinfo {author} {\bibfnamefont {A.}~\bibnamefont {Recati}},\ }\href
  {http://stacks.iop.org/1367-2630/11/i=10/a=103055} {\bibfield  {journal}
  {\bibinfo  {journal} {New Journal of Physics}\ }\textbf {\bibinfo {volume}
  {11}},\ \bibinfo {pages} {103055} (\bibinfo {year} {2009})}\BibitemShut
  {NoStop}%
\end{thebibliography}

%
\end{document}